\documentclass[pra,aps,nopacs,onecolumn,twoside,superscriptaddress]{revtex4}

\usepackage{amsmath,amsfonts,amssymb,caption,color,epsfig,graphics,graphicx,latexsym,mathrsfs,revsymb,theorem,url,verbatim,epstopdf}

\usepackage{hyperref}

\hypersetup{colorlinks,linkcolor={blue},citecolor={blue},urlcolor={red}} 

\usepackage{overpic}

\usepackage{hyperref}

\newtheorem{definition}{Definition}
\newtheorem{proposition}[definition]{Proposition}
\newtheorem{lemma}[definition]{Lemma}

\newtheorem{theorem}[definition]{Theorem}
\newtheorem{corollary}[definition]{Corollary}
\newtheorem{conjecture}[definition]{Conjecture}

\newtheorem{remark}[definition]{Remark}
\newtheorem{example}[definition]{Example}
\newtheorem{question}[definition]{Question}

\def\squareforqed{\hbox{\rlap{$\sqcap$}$\sqcup$}}
\def\qed{\ifmmode\squareforqed\else{\unskip\nobreak\hfil
		\penalty50\hskip1em\null\nobreak\hfil\squareforqed
		\parfillskip=0pt\finalhyphendemerits=0\endgraf}\fi}
\def\endenv{\ifmmode\;\else{\unskip\nobreak\hfil
		\penalty50\hskip1em\null\nobreak\hfil\;
		\parfillskip=0pt\finalhyphendemerits=0\endgraf}\fi}
\newenvironment{proof}{\noindent \textbf{{Proof.~} }}{\qed}
\def\Dbar{\leavevmode\lower.6ex\hbox to 0pt
	{\hskip-.23ex\accent"16\hss}D}
\makeatletter
\def\url@leostyle{%
	\@ifundefined{selectfont}{\def\UrlFont{\sf}}{\def\UrlFont{\small\ttfamily}}}
\makeatother
\def\bcj{\begin{conjecture}}
	\def\ecj{\end{conjecture}}
\def\bcr{\begin{corollary}}
	\def\ecr{\end{corollary}}
\def\bd{\begin{definition}}
	\def\ed{\end{definition}}
\def\bea{\begin{eqnarray}}
	\def\eea{\end{eqnarray}}
\def\bem{\begin{enumerate}}
	\def\eem{\end{enumerate}}
\def\bex{\begin{example}}
	\def\eex{\end{example}}
\def\bim{\begin{itemize}}
	\def\eim{\end{itemize}}
\def\bl{\begin{lemma}}
	\def\el{\end{lemma}}
\def\bma{\begin{bmatrix}}
	\def\ema{\end{bmatrix}}
\def\bpf{\begin{proof}}
	\def\epf{\end{proof}}
\def\bpp{\begin{proposition}}
	\def\epp{\end{proposition}}
\def\bqu{\begin{question}}
	\def\equ{\end{question}}
\def\br{\begin{remark}}
	\def\er{\end{remark}}
\def\bt{\begin{theorem}}
	\def\et{\end{theorem}}

\def\btb{\begin{tabular}}
	\def\etb{\end{tabular}}

\newcommand{\nc}{\newcommand}

\def\a{\alpha}
\def\b{\beta}
\def\g{\gamma}

\def\p{\pi}
\def\r{\rho}

\nc{\bbA}{\mathbb{A}} \nc{\bbB}{\mathbb{B}} \nc{\bbC}{\mathbb{C}}
\nc{\bbD}{\mathbb{D}} \nc{\bbE}{\mathbb{E}} \nc{\bbF}{\mathbb{F}}
\nc{\bbG}{\mathbb{G}} \nc{\bbH}{\mathbb{H}} \nc{\bbI}{\mathbb{I}}
\nc{\bbJ}{\mathbb{J}} \nc{\bbK}{\mathbb{K}} \nc{\bbL}{\mathbb{L}}
\nc{\bbM}{\mathbb{M}} \nc{\bbN}{\mathbb{N}} \nc{\bbO}{\mathbb{O}}
\nc{\bbP}{\mathbb{P}} \nc{\bbQ}{\mathbb{Q}} \nc{\bbR}{\mathbb{R}}
\nc{\bbS}{\mathbb{S}} \nc{\bbT}{\mathbb{T}} \nc{\bbU}{\mathbb{U}}
\nc{\bbV}{\mathbb{V}} \nc{\bbW}{\mathbb{W}} \nc{\bbX}{\mathbb{X}}
\nc{\bbZ}{\mathbb{Z}}

\nc{\bA}{{\bf A}} \nc{\bB}{{\bf B}} \nc{\bC}{{\bf C}}
\nc{\bD}{{\bf D}} \nc{\bE}{{\bf E}} \nc{\bF}{{\bf F}}
\nc{\bG}{{\bf G}} \nc{\bH}{{\bf H}} \nc{\bI}{{\bf I}}
\nc{\bJ}{{\bf J}} \nc{\bK}{{\bf K}} \nc{\bL}{{\bf L}}
\nc{\bM}{{\bf M}} \nc{\bN}{{\bf N}} \nc{\bO}{{\bf O}}
\nc{\bP}{{\bf P}} \nc{\bQ}{{\bf Q}} \nc{\bR}{{\bf R}}
\nc{\bS}{{\bf S}} \nc{\bT}{{\bf T}} \nc{\bU}{{\bf U}}
\nc{\bV}{{\bf V}} \nc{\bW}{{\bf W}} \nc{\bX}{{\bf X}}
\nc{\bZ}{{\bf Z}}

\nc{\cA}{{\cal A}} \nc{\cB}{{\cal B}} \nc{\cC}{{\cal C}}
\nc{\cD}{{\cal D}} \nc{\cE}{{\cal E}} \nc{\cF}{{\cal F}}
\nc{\cG}{{\cal G}} \nc{\cH}{{\cal H}} \nc{\cI}{{\cal I}}
\nc{\cJ}{{\cal J}} \nc{\cK}{{\cal K}} \nc{\cL}{{\cal L}}
\nc{\cM}{{\cal M}} \nc{\cN}{{\cal N}} \nc{\cO}{{\cal O}}
\nc{\cP}{{\cal P}} \nc{\cQ}{{\cal Q}} \nc{\cR}{{\cal R}}
\nc{\cS}{{\cal S}} \nc{\cT}{{\cal T}} \nc{\cU}{{\cal U}}
\nc{\cV}{{\cal V}} \nc{\cW}{{\cal W}} \nc{\cX}{{\cal X}}
\nc{\cZ}{{\cal Z}}

\nc{\hA}{{\hat{A}}} \nc{\hB}{{\hat{B}}} \nc{\hC}{{\hat{C}}}
\nc{\hD}{{\hat{D}}} \nc{\hE}{{\hat{E}}} \nc{\hF}{{\hat{F}}}
\nc{\hG}{{\hat{G}}} \nc{\hH}{{\hat{H}}} \nc{\hI}{{\hat{I}}}
\nc{\hJ}{{\hat{J}}} \nc{\hK}{{\hat{K}}} \nc{\hL}{{\hat{L}}}
\nc{\hM}{{\hat{M}}} \nc{\hN}{{\hat{N}}} \nc{\hO}{{\hat{O}}}
\nc{\hP}{{\hat{P}}} \nc{\hR}{{\hat{R}}} \nc{\hS}{{\hat{S}}}
\nc{\hT}{{\hat{T}}} \nc{\hU}{{\hat{U}}} \nc{\hV}{{\hat{V}}}
\nc{\hW}{{\hat{W}}} \nc{\hX}{{\hat{X}}} \nc{\hZ}{{\hat{Z}}}

\nc{\hn}{{\hat{n}}}

\def\diag{\mathop{\rm diag}}

\def\ghz{\mathop{\rm GHZ}}

\def\min{\mathop{\rm min}}

\def\tr{\mathop{\rm Tr}}

\newcommand{\bra}[1]{\langle#1|}
\newcommand{\ket}[1]{|#1\rangle}

\newcommand{\ketbra}[2]{|#1\rangle\!\langle#2|}

\def\Dbar{\leavevmode\lower.6ex\hbox to 0pt
	{\hskip-.23ex\accent"16\hss}D}

\begin{document}
\large

\title{Genuine entanglement, distillability and quantum information masking under noise}

\date{\today}


\author{Mengyao Hu}\email[]{mengyaohu@buaa.edu.cn}
\affiliation{School of Mathematical Sciences, Beihang University, Beijing 100191, China}

\author{Lin Chen}\email[]{linchen@buaa.edu.cn (corresponding author)}
\affiliation{School of Mathematical Sciences, Beihang University, Beijing 100191, China}
\affiliation{International Research Institute for Multidisciplinary Science, Beihang University, Beijing 100191, China}

\begin{abstract}
Genuineness and distillability of entanglement play a key role in quantum information tasks, and they are easily disturbed by the noise. We construct a family of multipartite states without genuine entanglement and distillability sudden death across every bipartition, respectively. They are realized by establishing the noise as the multipartite high dimensional Pauli channels. Further, we construct a locally unitary channel as another noise such that the multipartite Greenberger-Horne-Zeilinger state becomes the D{\"u}r's multipartite state. We also show that the quantum information masking still works under the noise we constructed, and thus show a novel quantum secret sharing scheme under noise. The evolution of a family of three-qutrit genuinely entangled states distillable across every bipartition under noise is also investigated.   
\end{abstract}

\maketitle

\section{Introduction}
Bipartite pure entangled states are genuinely entangled (GE) states \cite{PhysRevLett.86.5188}. They play a key role in quantum computing, secret sharing and superdense coding \cite{PhysRevA.77.032321, PhysRevA.97.062309, PhysRevLett.94.060501, PhysRevLett.125.230501, PhysRevLett.111.070501, 2018Experimental}. An explicit protocol has been proposed for faithfully teleporting an arbitrary two-qubit state by using a four-qubit GE state \cite{PhysRevLett.96.060502}. A $(2n+1)$-qubit GE state has been constructed to perform controlled teleportation of an arbitrary $n$-qubit state \cite{PhysRevA.75.052306}. However, Genuine entanglement may be influenced by the  unavoidable noise in environment. \textit{Genuine entanglement sudden death} (GESD) describes that GE states evolve into non-GE states under noise. So it is indispensable to sustain the genuine entanglement of quantum states under noise. We refer to the \textit{GESD-free states} as multipartite quantum states without GESD. We construct such states, and it is the first motivation of this paper.

A bipartite quantum state is entangled or separable \cite{PhysRevLett.77.1413,PhysRevLett.88.187904}. The phenomenon of a finite-time disappearance of entanglement under noise is known as entanglement sudden death (ESD).  For example, ESD is exactly GESD for bipartite states. Bipartite entangled states can be divided into free and bound types \cite{PhysRevLett.80.5239,RevModPhys.81.865}. Free entangled states can be asymptotically distilled into pure-state entanglement under local operations and classical communications (LOCC), while bound entangled states cannot. Only free entangled states are essential ingredients in quantum-information tasks. The notion of distillability sudden death (DSD) describes that free entangled states can evolve into nondistillable (bound entangled or separable) states in finite time under local noise \cite{DSD_SW,DSD_MA}. It is thus desirable to construct DSD-free states, i.e., states without DSD under noise. Actually, two-qutrit DSD-free states have been proposed in \cite{DSD_SW}. We extend the notion of DSD-free states from bipartite states to multipartite states in any dimension. We refer to the later as \textit{multipartite DSD-free states}. As far as we know, little is known about multipartite DSD-free states. We investigate a family of these states which go through the noise characterized by the multipartite high dimensional Pauli channel. This is the second motivation of this paper. 

The family of Pauli channels represents a wide class of noise process \cite{PhysRevLett.98.030501}. The experimental realization of the optimal estimation protocol for a Pauli noise channel has been constructed \cite{PhysRevLett.107.253602}. Though some noise channels are not Pauli channels, they can be well approximated by Pauli channels without introducing new errors \cite{quantumnoise2020, PhysRevA.94.052325}. Quantum noise is connected to quantum measurement and amplification \cite{RevModPhys.82.1155, PhysRevA.59.4249}. Thus it is crucial to characterize quantum noise channels reliably and efficiently. Recently, a protocol is described and implemented experimentally on a 14-qubit superconducting quantum architecture in terms of an estimation of the effective noise over qubit systems \cite{naturenoise}. It thus applies to well-known multiqubit states such as the Greenberger-Horne-Zeilinger (GHZ) state \cite{PhysRevA.78.022110}. It is widely used to split quantum information in secret sharing \cite{PhysRevA.59.1829} and works as a channel in quantum teleportation \cite{DAI2004281, PhysRevLett.96.060502}, as well as quantum information masking in bipartite systems \cite{PhysRevLett.120.230501}. Moreover, quantum information masking plays a key role in quantum secret sharing \cite{PhysRevLett.83.648, PhysRevA.69.052307, PhysRevA.98.062306}. This concept has been extended to $m$-uniform quantum information masking in multipartite systems \cite{feishimask}. This kind of masking allows quantum secret sharing from a "boss" to his "subordinates" since the $m$ subordiantes cannot retrieve the information even if they collaborate.  However, the existing schemes of quantum information masking work for pure states only, and they should be investigated under noise from a practical point of view.

In this paper, we construct a family of GESD-free states and multipartite DSD-free states from GHZ states. We begin with reviewing the definition of GE state, and introduce the multipartite DSD-free states in Definition \ref{def:sDSD-free}. We present the realignment criterion in Lemma \ref{le:realignment} to detect entanglement. We review the fact from Ref. \cite{PhysRevLett.104.210501} about GE states coupled with white noise in Lemma \ref{le:GHZ+I}. It shows that the state $\eta:=p\ketbra{{\rm GHZ}_{d,n}}{{\rm GHZ} _{d,n}}+(1-p)\frac{1}{d^n}\mathbb{I}_{d^n}$ is GE if $p>\frac{3}{d^{n-1}+3}$. In Theorem \ref{thm:GESD_channel} we construct the $n$-partite $d$-dimensonal Pauli channel such that the $d$-level $n$-partite GHZ state becomes the state $\eta$. We show that $\eta$ is GESD-free if $p \in (\frac{3}{d^{n-1}+3},1]$. We prove that the state $\eta$ is distillable across every bipartition, that is, multipartite DSD-free, if and only if $p \in (\frac{1}{1+d^{n-1}},1]$ in Lemma \ref{le:distill_bipartition}. Next in Theorem \ref{thm:ghzBE}, we show that the state $\eta$ is not at the same time GE and positive partial transpose (PPT) across any given bipartition. In Theorem \ref{thm:channelBEFE}, we construct the locally unitary channel such that the $N$-partite GHZ state becomes the D{\"u}r's multipartite state \cite{PhysRevLett.87.230402, Zhu_2010, PhysRevA.70.022322}. The condition for the biseparability of such a state is given in Theorem \ref{thm:dursep}. We also study the quantum information masking under noise in terms of the state $\eta$ and D{\"u}r's multipartite state. Furthermore, we investigate the evolution of a three-qutrit GE state which is distillable across every bipartition under global (i.e., local and collective) noise. The numerical results show that the state will undergo DSD and become PPT.   

The noise characterized by the channels proposed in this paper can be investigated more. These channels are operational in experiment \cite{quantumnoise2020, K2005Distributing}. Many noisy quantum channels in nature are a convex combination of unitary channels \cite{girard2020mixed}. It is important to characterize these channels reliably and efficiently with high precision since it is the central obstacle in quantum computation and secret sharing. The quantum information masking under noise we proposed is capable of being applied in secret sharing and future quantum communication protocols.

The rest of this paper is organized as follows. In Sec. \ref{sec:pre}, we introduce the preliminary knowledge used in this paper. In Sec.  \ref{sec:GESD}, we construct a family of GESD-free states from GHZ states by using the high dimensional Pauli channels. In Sec. \ref{sec:multipartiteDSD-FREE}, we investigate the multipartite DSD-free states constructed from GHZ states. In Sec. \ref{sec:masking}, we investigate quantum information masking under noise. In Sec. \ref{sec:DSD333}, we give an example of a GE state distillable across every bipartition under local and collective noise. Finally, we conclude in Sec. \ref{sec:con}.  

\section{Preliminaries}
\label{sec:pre}
In this section, we introduce the fundamental knowledge used throughout the paper. Let $\bbC^d$ be the $d$-dimensional Hilbert space and $\cH_{A}$ be the  Hilbert space of system A. Given a bipartite state $\r$ on $\cH_{AB}$, the partial trace with regard to system A is defined as $\r_B:=\text{Tr}_A \r:=\sum_i(\bra{a_i}\otimes I_B)\r(\ket{a_i}\otimes I_B):=\sum_i\bra{a_i}_A \r \ket{a_i}_A$, where $\ket{a_i}$ is an arbitrary orthonormal basis in $\cH_{A}$. We call $\r_B$ the reduced density matrix of system B. Next, we review the GE states \cite{PhysRevLett.102.250404, PhysRevLett.104.210501}. Let $\ket{\psi_n}$ be an $n$-partite pure state in the Hilbert space $\cH_{A_1A_2 \dots A_n}=\cH_{A_1}\otimes \cH_{A_2} \otimes \dots \otimes \cH_{A_n}=\bbC^{d_1}\otimes \bbC^{d_2} \otimes \dots \otimes \bbC^{d_n}$. Then $\ket{\psi_n}$ is an $n$-partite GE pure state if $\ket{\psi_j} \neq \ket{a,b} \in \cH_{A_{j_1\dots j_k}}\otimes \cH_{A_{j_{k+1}\dots j_n}}$, 
	where $\{j_1,\dots,j_n\}=\{1,\dots,n\}$. Further, we say that $\r \in \cB(\cH_{A_1A_2 \dots A_n})$ is a mixed $n$-partite GE state if 
\begin{eqnarray}
	\label{eq:mixedGE}
	\r=\sum_{j=1}^r p_j \ketbra{\psi_j}{\psi_j},  
\end{eqnarray}
where $\ket{\psi_j}$ is a pure $n$-partite GE state in every decomposition, $\sum_{j=1}^r p_j=1, p_j\geq 0$. Genuine entanglement sudden death (GESD) describes that GE states evolve into biseparable states, that is, it cannot be written as the form in Eq. \eqref{eq:mixedGE}. In Sec. \ref{sec:GESD}, we shall construct  GESD-free states by using GHZ states.

Next, we present the definition of multipartite DSD-free states. 
\begin{definition}
	\label{def:sDSD-free}
	Multipartite DSD-free states are $n$-partite quantum states without DSD in every bipartition. They are distillable and non-positive partial transpose (NPT) \footnote{If the partial transpose with regard to one system of a bipartite state has at least one negative eigenvalue, then the state has non-positive partial transpose} across every bipartition. 
\end{definition}

Such multipartite DSD-free states will be constructed in Sec. \ref{sec:multipartiteDSD-FREE}. DSD-free states always have NPT. It has been proven that a bipartite state has PPT is non-distillable under LOCC \cite{PhysRevLett.80.5239,DSD_SW}. So PPT entangled states must be bound entangled \cite{horodecki2001distillation,halder2019construction}. However, there are evidences showing that bound entangled states with NPT exist \cite{PhysRevA.61.062313,PhysRevA.61.062312}. It will be investigated in Sec. \ref{sec:DSD333}.

In the following, we introduce the realignment criterion proposed in \cite{Chen2002A}. It provides a necessary criterion for separable states \cite{Chen2002A,CHEN200214} and a method of detecting BE states \cite{RevModPhys.81.865}. We only consider bound entangled states with PPT. Realignment criterion provides a computable necessary criterion to detect separability \cite{Chen2002A}. It will be used to detect entanglement in Sec. \ref{sec:DSD333}.
\begin{lemma}
	\label{le:realignment}
	If an $m \times n$ bipartite density matrix $\r$ is separable, then for the $m^2 \times n^2$ matrix  $(\r)^R_{ij,kl}=\r_{ij,kl}$, one can obtain that $\|\r^R\|\leq 1$.
\end{lemma}

\section{GESD-FREE States}
\label{sec:GESD}

In this section, we construct a family of GESD-free states in terms of the well-known GHZ state which goes through the multipartite high dimensional Pauli channels. The first main result of this section is Theorem \ref{thm:GESD_channel}. This is supported by Lemma \ref{le:GHZ+I}. We present the keys of the three-qubit state in Example \ref{ex:3qubit}. This is a more implementable case due to the practically realizable $n$-qubit states up to ten \cite{PhysRevA.101.062330, PhysRevX.9.031045, PhysRevLett.117.210502}.

One can show that the pure GE state is NPT across every bipartition.
An example is the $d$-level $n$-partite GHZ state
\begin{eqnarray}
	\label{eq:GHZ}
	\ket{{\rm GHZ}_{d,n}}:={1\over \sqrt d}\sum^{d-1}_{j=0}\ket{j,j,...,j}.
\end{eqnarray}
In practice, the GHZ state will be coupled with white noise. We need to know whether the resulting state is GE. The following fact is from Example 4 in \cite{PhysRevLett.104.210501}. 
\begin{lemma}
	\label{le:GHZ+I}
	Consider the GHZ state with additional isotropic (white) noise,
	\begin{eqnarray}
		\label{eq:GHZ+I}
		\eta:=p\ketbra{{\rm GHZ}_{d,n}}{{\rm GHZ} _{d,n}}+(1-p)\frac{1}{d^n}\mathbb{I}_{d^n}. 
	\end{eqnarray}
	Then $\eta$ is GE for $p>\frac{3}{d^{n-1}+3}$.
\end{lemma}

We construct a locally unitary channel $\Delta$  such that $\ket{\ghz_{d,n} }$ becomes the state $\eta$ in Eq. \eqref{eq:GHZ+I} when $d=2$ and $n=3$. 
\begin{example}
	\label{ex:3qubit}
	Let $d=2$ and $n=3$, that is, $\ket{\ghz_{2,3} }:={1\over \sqrt 2}\sum^{1}_{j=0}\ket{j,j,j}$. We denote the Pauli matrices as follows,
	\begin{eqnarray}
		\label{ma:Pauli}
		\sigma_x=\bma
		0 & 1 \\
		1 & 0 
		\ema, \quad 
		\sigma_y=\bma
		0 & -i \\
		i & 0
		\ema, \quad
		\sigma_z=\bma
		1 & 0\\
		0 & -1
		\ema.
	\end{eqnarray} 
	Let the Kraus operators be
	\begin{eqnarray}
		\begin{aligned}
			&	K_1=\mathbb{I}_2 \otimes \mathbb{I}_2 \otimes \mathbb{I}_2, & K_2=\mathbb{I}_2 \otimes \mathbb{I}_2 \otimes \sigma_z,\\ \notag
			& K_3=\mathbb{I}_2 \otimes \mathbb{I}_2 \otimes \sigma_x, & K_4=\mathbb{I}_2 \otimes \sigma_z \otimes \sigma_x, \\ \notag
			& K_5=\mathbb{I}_2 \otimes \sigma_x \otimes \mathbb{I}_2,
			& K_6=\mathbb{I}_2 \otimes \sigma_x \otimes \sigma_z, \\ \notag
			& K_7=\sigma_x \otimes \mathbb{I}_2 \otimes \mathbb{I}_2, 
			& K_8=\sigma_x \otimes \mathbb{I}_2 \otimes \sigma_z. 
		\end{aligned}
	\end{eqnarray}
	Let the locally unitary channel $\Delta (\cdot)=pK_1(\cdot)K_1^\dagger+\frac{1-p}{8}\sum_{i=1}^8K_i(\cdot)K_i^\dagger$. We refer to it as the tripartite Pauli channel. One can verify that 
	\begin{eqnarray}
		&&
		\Delta(\ketbra{{\rm GHZ}_{2,3} }{{\rm GHZ}_{2,3} }) \notag \\
		=&&
		pK_1\ketbra{{\rm GHZ}_{2,3} }{{\rm GHZ}_{2,3} }K_1^\dagger+\frac{1-p}{8}\sum_{i=1}^8K_i\ketbra{{\rm GHZ}_{2,3} }{{\rm GHZ}_{2,3} }K_i^\dagger
	\end{eqnarray}  
	is exactly the state in \eqref{eq:GHZ+I} with $d=2$ and $n=3$.
	From Lemma \ref{le:GHZ+I}, we know that $\eta$ is GE if $p \in (\frac{3}{7},1)$.
	\qed
\end{example}

We extend the above channel $\Delta$ to a channel of any $d$ and $n$. We refer to it as the $n$-partite $d$-dimensional Pauli channel. In this position, we present the main result of this section. 

\begin{theorem}
	\label{thm:GESD_channel}
	There exists an $n$-partite $d$-dimensional Pauli channel such that $\ket{\ghz_{d,n} }$ becomes the state $\eta$ in \eqref{eq:GHZ+I}. It is GESD-free when $p \in (\frac{3}{d^{n-1}+3},1]$.
\end{theorem}
\begin{proof}
	Let 
	\begin{eqnarray}
		\label{eq:basisa}
		&&
		\ket{a_{i_1,i_2,\cdots ,i_{n-1}}^q}=\frac{1}{\sqrt{d}}\sum_{j=0}^{d-1}e^{\frac{2\p i}{d}jq}\ket{j,j\oplus i_1,j\oplus i_2,\cdots,j\oplus i_{n-1}},
	\end{eqnarray}
	where $j\oplus i=(j+i) \mod d$, and $q,j,i_1,i_2,\cdots,i_{n-1}=0,1,\cdots,d-1$.
	By using \eqref{eq:basisa}, the $d^n$-dimensional identity matrix can be expressed as
	\begin{eqnarray}
		\label{eq:Idn}
		&&
		\mathbb{I}_{d^n}=\sum_{q,i_1,i_2,\cdots,i_{n-1}=0}^{d-1} \ketbra{a_{i_1,i_2,\cdots,i_{n-1}}^q}{a_{i_1,i_2,\cdots,i_{n-1}}^q}.
	\end{eqnarray}

 For $d$-level system, a set of local orthogonal observables (LOOs)  \cite{PhysRevLett.95.150504} is a set of $d^2$ observables $A_\mu (\mu=1,2,\cdots, d^2)$ satisfying orthogonal relations $\tr (A_\mu A_\nu)=\delta_{\mu \nu}$, where $\mu,\nu=1,2,\cdots,d^2$. For example, a standard complete set of LOOs is defined as $\{A_\mu\}=\{A_m=\ketbra{m}{m},A_{m,n}^{\pm}\}$, where
	\begin{eqnarray}
	\label{eq:loo}
	&&
	A_{m,n}^{+}=\frac{\ketbra{m}{n}+\ketbra{n}{m}}{\sqrt{2}} \quad (m<n), \\
	&&
	A_{m,n}^{-}=\frac{\ketbra{m}{n}-\ketbra{n}{m}}{i\sqrt{2}} \quad (m<n), \notag
	\end{eqnarray}
	and $m,n=0,1,\cdots,d-1$. For convenience, we denote $A_{m,n}^+$ in Eq. \eqref{eq:loo} as  
	\begin{eqnarray}
	\label{eq:loo+}
	\begin{aligned}
	& A_{j\oplus i,j}^+=\sqrt{2}A_{j\oplus i,j}^+ \quad & ( j\oplus i<j ), \\
	& A_{j\oplus i,j}^+=\sqrt{2}A_{j,j\oplus i}^+ \quad & ( j\oplus i>j),\\
	& A_{j\oplus i,j}^+=\mathbb{I}_d, \quad & ( j\oplus i=j),
	\end{aligned}
	\end{eqnarray}
Let the Fourier matrix $W=\frac{1}{\sqrt{d}}[e^{\frac{2\p i}{d}jq}]_{j,q}$. Suppose the Kraus operators are 
	\begin{eqnarray}
	\label{eq:KrausOperators}
	&&
	K_0= \mathbb{I}_d\otimes \mathbb{I}_d \otimes \cdots \otimes \mathbb{I}_d,  \notag  \\ 
	&&
	K_r=W\otimes A_{j\oplus i_1,j}^+\otimes \cdots \otimes A_{j\oplus i_{n-1},j}^+.
	\end{eqnarray}
	where $r=1,2,\cdots,d^n$. Then using \eqref{eq:Idn} and $\eqref{eq:KrausOperators}$, one can verify that $\sum_{r=1}^{d^n}K_r\ketbra{{\rm GHZ}_{d,n} }{{\rm GHZ}_{d,n} }K_r^\dagger=\mathbb{I}_{d^n}$.	We define the $n$-partite $d$-dimensional Pauli channel $\cU$ as
	\begin{eqnarray}
	\label{eq:channelU}
	\cU (\cdot):=pK_0(\cdot)K_0^\dagger+\frac{1-p}{d^n}\sum_{r=1}^{d^n}K_r(\cdot)K_r^\dagger.
	\end{eqnarray}
	We have 
	\begin{eqnarray}
		\label{eq:r-channel}
		&&
		\cU(\ketbra{{\rm GHZ}_{d,n} }{{\rm GHZ}_{d,n} }) \notag \\
		=&&pK_0\ketbra{{\rm GHZ}_{d,n} }{{\rm GHZ}_{d,n} }K_0^\dagger+\frac{1-p}{d^n}\sum_{r=1}^{d^n}K_r\ketbra{{\rm GHZ}_{d,n} }{{\rm GHZ}_{d,n} }K_r^\dagger. 
	\end{eqnarray}
	One can verify that the state is exactly the state $\eta$ in Eq. \eqref{eq:GHZ+I}. From Lemma \ref{le:GHZ+I}, we know that $\eta$ is GE if $p \in (\frac{3}{d^{n-1}+3},1]$. Thus we can obtain that when $p \in (\frac{3}{d^{n-1}+3},1]$, the state in \eqref{eq:r-channel} is GESD-free. This completes the proof. 	
\end{proof}

\begin{figure}[!h] 
	\center{\includegraphics[width=14cm]  {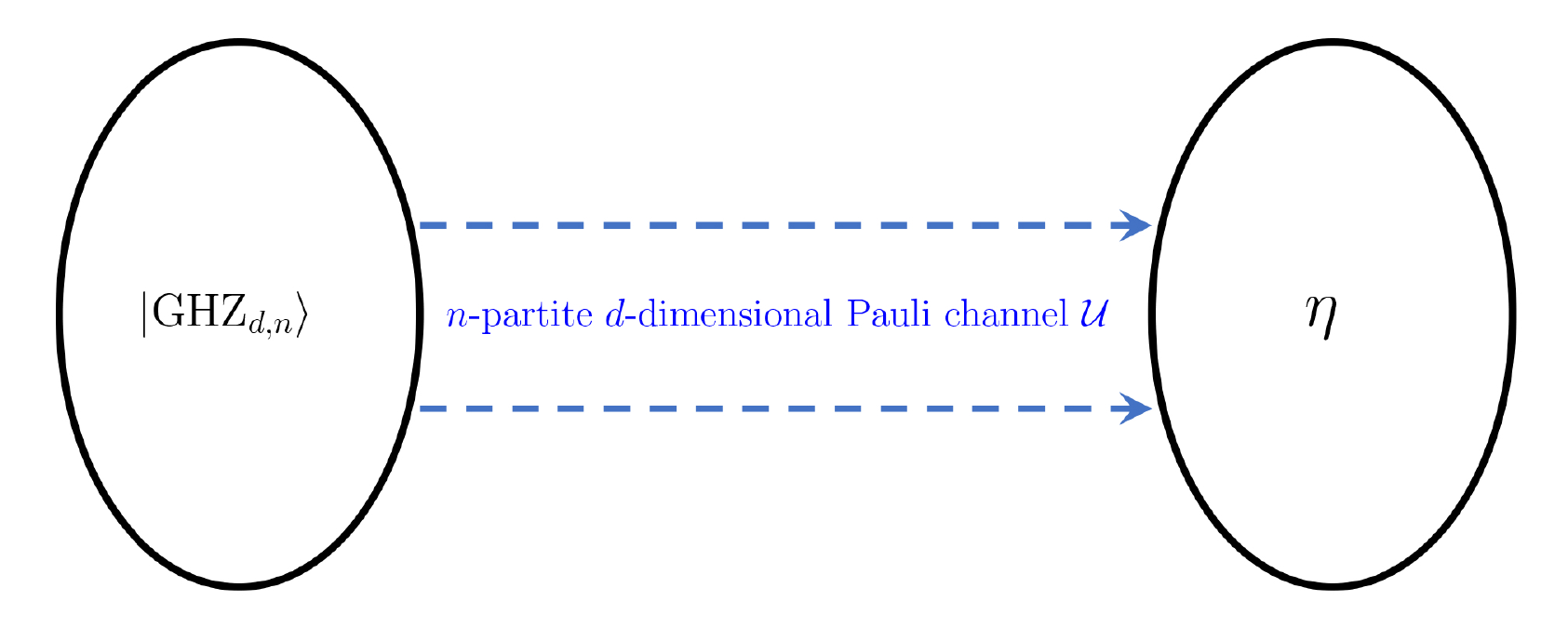}} 
	\caption{\label{fig:channel} The state $\ket{{\rm GHZ}_{d,n}}$ goes through the $n$-partite $d$-dimensional Pauli channel $\cU$ in \eqref{eq:r-channel}, and becomes the state $\eta$ in \eqref{le:GHZ+I}. 
	}  
\end{figure}

We illustrate Theorem \ref{thm:GESD_channel} by FIG. \ref{fig:channel}. By selecting $p \in (\frac{3}{d^{n-1}+3},1]$, we can construct such GESD-free state $\eta$ in experiment. When $d$ or $n$ is large enough, then $\frac{3}{d^{n-1}+3} \rightarrow 0$. Under these circumstances, the state $\eta$ is almost always GE. On the other hand, there has been substantial progress towards the realization of GHZ states and noise in experiment. Recently, the ten-photon GHZ state in experiment has been demonstrated \cite{PhysRevLett.117.210502}. Further, a three-particle GHZ state entangled in three levels for every particle has been realized in 2018 \cite{2018Experimental}. By choosing local basis rotations properly, a complete basis of GHZ states can be constructed. The four qubit GHZ states have been demonstrated experimentally by entangling two photons in polarization and orbital angular momentum \cite{2017Experimental}. Moreover, the multipartite high dimensional Pauli channel in this paper is constructed from LOOs. One can see that they are constructed by the computational basis in $\bbC^d$. Based on LOOs, Ref \cite{PhysRevLett.95.150504} also proposed a family of entanglement witness and corresponding positive maps that are not completely positive. Thus, we provide an operational way to construct a family of GE states in \eqref{eq:r-channel} by the $d$-level $n$-partite GHZ state through locally unitary channels in experiment.

\section{Multipartite DSD-FREE States}
\label{sec:multipartiteDSD-FREE}

In this section, we construct the multipartite DSD-free states under white noise in Lemma \ref{le:distill_bipartition}. We also investigate their genuine entanglement and distillability in Theorem \ref{thm:ghzBE}. It is supported by Lemmas \ref{le:GHZ+I} and \ref{le:distill_bipartition}. Next in Theorem \ref{thm:channelBEFE}, we construct the locally unitary channel such that the  $N$-partite GHZ state becomes the D{\"u}r's multipartite state. Further, the condition for the biseparability of the D{\"u}r's multipartite state is given in Theorem \ref{thm:dursep}.

First, we begin with the distillability of the state $\eta$ in \eqref{eq:GHZ+I}. It will be used in the proof of Theorem \ref{thm:ghzBE}.
\begin{lemma}
	\label{le:distill_bipartition}
	The state $\eta$
	is distillable across every bipartition if and only if $p \in (\frac{1}{1+d^{n-1}},1]$.
\end{lemma}
\begin{proof}
	For the state  $\eta_{A_1\cdots A_n}=p\ketbra{\ghz_{d,n} }{\ghz_{d,n} }+(1-p)\frac{1}{d^n}\mathbb{I}_{d^n}$, one can obtain that this $k$-partite reduced density matrices $\eta_{A_{j_1}\cdots A_{j_k}}$ of $\eta_{A_1\cdots A_n}$ for any given $k \in \{1,\cdots,n-1\}$ are the same, where $j_1,\cdots,j_k \in \{1,\cdots,n\}$. Thus it suffices to show that the state $\eta_{A_1 \cdots A_n}$ is distillable across any one of the bipartite cuts $\eta_{A_1|A_2\cdots A_n}, \eta_{A_1A_2|A_3\cdots A_n},\cdots,\eta_{A_1\cdots A_{\lfloor\frac{n}{2}\rfloor}|A_{\lfloor\frac{n}{2}\rfloor+1} \cdots A_n}$.
	
	Let $P=\sum_{i=0}^1 \ketbra{\underbrace{i,\cdots,i}_{k}}{\underbrace{i,\cdots,i}_{k}}\otimes \sum_{j=0}^1 \ketbra{\underbrace{j,\cdots,j}_{n-k}}{\underbrace{j,\cdots,j}_{n-k}} $, where $k=1,2,\cdots,\lfloor\frac{n}{2}\rfloor$. Then we can obtain that
	\begin{eqnarray}
		\label{eq:prhop}
		P(\eta_{A_1\cdots A_n})P^\dagger=&&\frac{p}{d}(\sum_{i=0}^{1}\ket{\underbrace{i,\cdots,i}_n})(\sum_{j=0}^{1}\bra{\underbrace{j,\cdots,j}_n}) \notag \\
		&&+\frac{1-p}{d^n}\sum_{i=0}^1\ketbra{\underbrace{i,\cdots,i}_{k}}{\underbrace{i,\cdots,i}_{k}}\otimes \sum_{j=0}^1\ketbra{\underbrace{j,\cdots,j}_{n-k}}{\underbrace{j,\cdots,j}_{n-k}}.
	\end{eqnarray}
	
	It is a state on $\bbC^2 \otimes \bbC^{2^{n-1}}$. Using the fact that any $2 \otimes N (N \geq 2)$ states are distillable under LOCC if and only if they are NPT \cite{PhysRevLett.78.574}, it suffices to show that the state $P(\eta_{A_1,\dots,A_n})P^\dagger$ in \eqref{eq:prhop} is NPT across any one of the bipartite cuts $\eta_{A_1|A_2\dots A_n}, \eta_{A_1A_2|A_3\dots A_n},\dots,\eta_{A_1\dots A_{\lfloor\frac{n}{2}\rfloor}|A_{\lfloor\frac{n}{2}\rfloor+1} \dots A_n}$.
	
	First we prove the "only if" part. Suppose that $P(\eta_{A_1\cdots A_n})P^\dagger$ is NPT. The partial transpose with regard to system $A_1\cdots A_k$ is
	\begin{eqnarray}
		\label{eq:prhop_trans}
		(P(\eta_{A_1\cdots A_n})P^\dagger)^\Gamma=&&\frac{p}{d}\sum_{i=0}^{1} \sum_{j=0}^{1} \ketbra{\underbrace{i,\cdots,i}_{k},\underbrace{j,\cdots,j}_{n-k}}{\underbrace{j,\cdots,j}_{k},\underbrace{i,\cdots,i}_{n-k}}\notag \\
		&&+\frac{1-p}{d^n}\sum_{i=0}^1\ketbra{\underbrace{i,\cdots,i}_{k}}{\underbrace{i,\cdots,i}_{k}}\otimes \sum_{j=0}^1\ketbra{\underbrace{j,\cdots,j}_{n-k}}{\underbrace{j,\cdots,j}_{n-k}}.
	\end{eqnarray}
	It implies that if $P(\eta_{A_1\cdots A_n})P^\dagger$ is NPT, then $-\frac{p}{d}+(1-p)\frac{1}{d^n}<0$. So we have  $p>\frac{1}{1+d^{n-1}}$.
	
	Second we prove the "if" part. Suppose $p>\frac{1}{1+d^{n-1}}$. Then from \eqref{eq:prhop_trans}, we can obtain that the minimum eigenvalue of $(P(\eta_{A_1\cdots A_n})P^\dagger)^\Gamma$ is negative. 
	
	Therefore, the state $p\ketbra{\ghz_{d,n} }{\ghz_{d,n} }+(1-p)\frac{1}{d^n}\mathbb{I}_{d^n}$
	is distillable across every bipartition if and only if $p \in (\frac{1}{1+d^{n-1}},1]$. This completes the proof.
\end{proof}

Hence, when the state $\ket{{\rm GHZ}_{d,n}}$ goes through the $n$-partite $d$-dimensional Pauli channel $\cU$ in \eqref{eq:r-channel}, the output state $\eta$ is multipartite DSD-free if and only if $p \in (\frac{1}{1+d^{n-1}},1]$. Now we are in the position to prove the main result of this section.

\begin{theorem}
	\label{thm:ghzBE}
The state $\eta$ is not at the same time GE and PPT across any given bipartition.
\end{theorem}
\begin{proof}
	From Lemma \ref{le:distill_bipartition}, we know that the state $\eta=p\ketbra{\ghz_{d,n} }{\ghz_{d,n}}+(1-p)\frac{1}{d^n}\mathbb{I}_{d^n}$ is NPT if and only if it is distillable, and it is distillable across every bipartition if and only if $p \in (\frac{1}{1+d^{n-1}},1]$. Thus it is PPT if and only if $p\leq \frac{1}{1+d^{n-1}}$. From Lemma \ref{le:GHZ+I}, we know that $\eta$ is GE if $p>\frac{3}{d^{n-1}+3}$. Note that  $\frac{3}{d^{n-1}+3}$ is larger than $\frac{1}{1+d^{n-1}}$. Thus $\eta$ cannot at the same time be GE and PPT across any given  bipartition. This completes the proof.
\end{proof}

Lemma \ref{le:distill_bipartition} and Theorem \ref{thm:ghzBE} imply that $\eta$ will always be GE and distillable across every bipartition. It provides an operational way in experiments to construct a family of this kind of states  \cite{PhysRevLett.117.210502, 2018Experimental, 2017Experimental}. We describe Theorem \ref{thm:GESD_channel} and Lemma \ref{le:distill_bipartition} in FIG. \ref{fig:GESD_DSD}. 

\begin{figure}[!h] 
	\center{\includegraphics[width=16cm]  {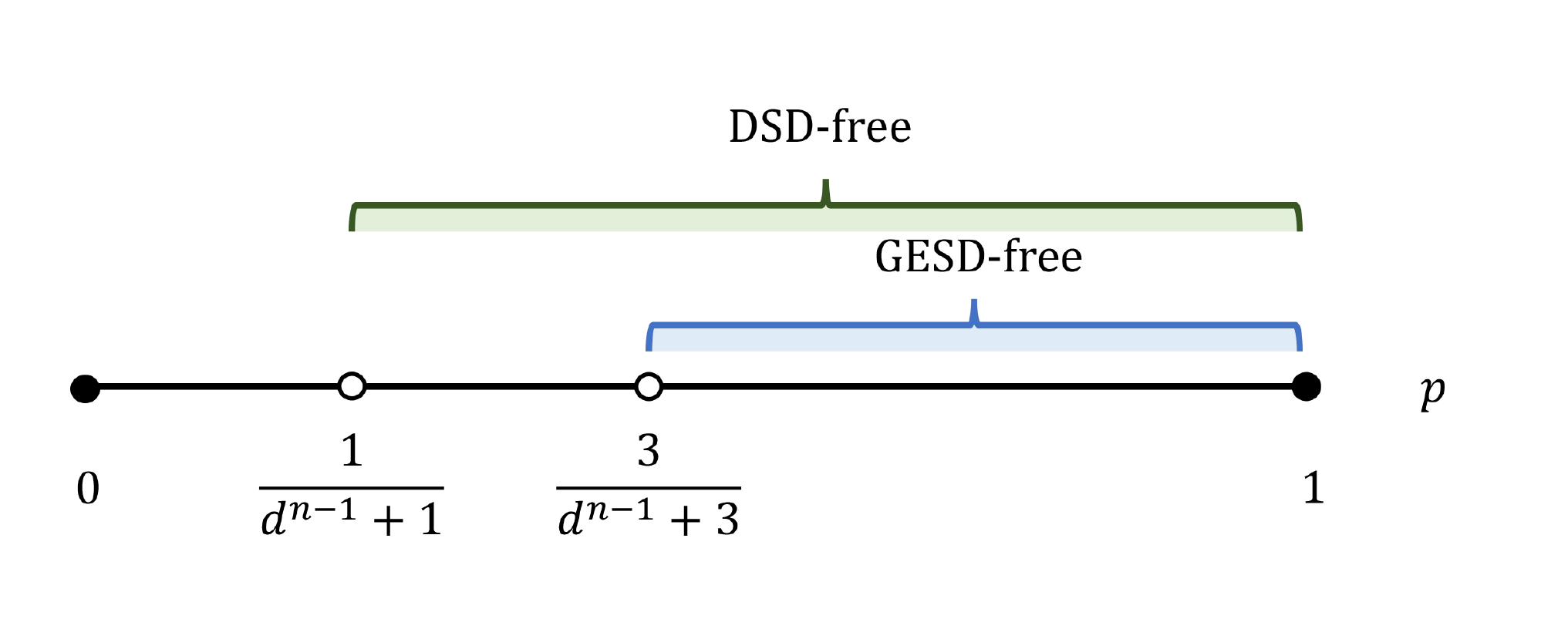}} 
	\caption{\label{fig:GESD_DSD} When the state $\ket{\ghz_{d,n} }$ goes through the locally unitary channel $\cU=pK_0(\cdot)K_0^\dagger+\frac{1-p}{d^n}\sum_{r=1}^{d^n}K_r(\cdot)K_r^\dagger$, the output state $\eta$ in \eqref{eq:GHZ+I} is GESD-free if $p \in (\frac{3}{d^{n-1}+3},1]$ and multipartite DSD-free if $p \in (\frac{1}{1+d^{n-1}},1]$, respectively.
	}  
\end{figure}

So far, we have investigated the GHZ states under white noise. Actually, they may undergo other noise and become other state. For example, the D{\"u}r's multipartite state \cite{PhysRevLett.87.230402, Zhu_2010, PhysRevA.70.022322}
\begin{eqnarray}
	\label{eq:rn}
\r_N(x)=x\ketbra{\Psi_G}{\Psi_G}+\frac{1-x}{2N}\sum_{k=1}^N(P_k+\overline{P}_k),
\end{eqnarray} 
where $\ket{\Psi_G}=\frac{1}{\sqrt{2}}(\ket{0^{\otimes N}}+e^{i\a_N}\ket{1^{\otimes N}})$ is the $N$-partite GHZ state, $P_k$ is the projector onto the product state $\ket{\phi_k}=\ket{0}_{A_1}\ket{0}_{A_2}\cdots\ket{1}_{A_k}\cdots \ket{0}_{A_N}$, and $\overline{P}_k$ is the projector onto the product state $\ket{\varphi_k}=\ket{1}_{A_1}\ket{1}_{A_2}\cdots\ket{0}_{A_k}\cdots \ket{1}_{A_N}$. We may take $e^{i\a_N}=1$ in $\ket{\Psi_G}$ since it can be eliminated by local unitary transformations. In the following we construct the locally unitary channel such that the  multipartite GHZ state becomes D{\"u}r's multipartite state $\r_N(x)$.

\begin{theorem}
	\label{thm:channelBEFE}
	There exists the channel $\Lambda$ such that the $N$-partite GHZ state $\ket{\Psi_G}$ becomes the state $\r_N(x)$ in \eqref{eq:rn}. The state $\r_N(x)$ is bound entangled if $0\leq x \leq \frac{1}{N+1}$ and free entangled if $\frac{1}{N+1} < x \leq 1$.
\end{theorem}
\begin{proof}
	Let the Kraus operators be 
	\begin{eqnarray}
	\label{eq:operatorK}
	&&
	K_0=\mathbb{I}_2 \otimes \cdots \otimes \mathbb{I}_2, \notag \\
	&&
	K_i=(\ket{0}_{A_1}\ket{0}_{A_2}\cdots\ket{1}_{A_k}\cdots \ket{0}_{A_N})(\bra{0}_{A_1}\bra{0}_{A_2}\cdots\bra{0}_{A_k}\cdots \bra{0}_{A_N}), \quad i=1,\cdots,N, \notag \\
	&&
	K_j=(\ket{1}_{A_1}\ket{1}_{A_2}\cdots\ket{0}_{A_k}\cdots \ket{1}_{A_N})(\bra{1}_{A_1}\bra{1}_{A_2}\cdots\bra{1}_{A_k}\cdots \bra{1}_{A_N}), \quad j=N+1,\cdots,2N, \notag \\
	&&
	K_{2N+1}=\sqrt{2N} \cdot K_0 -\sum_{i=1}^{2N}K_i,
	\end{eqnarray}
	where $k=1,\cdots,N$. Suppose the locally unitary channel
	\begin{eqnarray}
		\label{eq:channellambda}
		\Lambda (\cdot):=x K_0(\cdot)K_0^\dagger+\frac{1-x}{2N}\sum_{i=1}^{2N+1}K_i(\cdot)K_i^\dagger,
	\end{eqnarray}
	 where $K_i,i=0,\cdots,2N+1$ are the Kraus operators in \eqref{eq:operatorK}. One can verify that $x K_0^\dagger K_0+\frac{1-x}{2N}\sum_{i=1}^{2N+1}K_i^\dagger K_i=\mathbb{I}_{2^{N}}$. Using the channel $\Lambda$, we can obtain that $\Lambda(\ketbra{\Psi_G}{\Psi_G}) 		=x\ketbra{\Psi_G}{\Psi_G}+\frac{1-x}{2N}\sum_{k=1}^N(P_k+\overline{P}_k)$ is exactly the state $\r_N(x)$ in \eqref{eq:rn}. The second claim has been proved in \cite{PhysRevA.70.022322}. We have proven the assertion.
\end{proof}

This theorem implies that if we change the channel $\Lambda$ by $x$ from $x>\frac{1}{N+1}$ to $x\leq \frac{1}{N+1}$, then the state $\r_N(x)$ will undergo DSD. We can also obtain that when the $N$-partite GHZ state $\ket{\Psi_G}$ goes through the channel $\Lambda$, the output state $\r_N(x)$ is multipartite DSD-free if and only if $x \in (\frac{1}{N+1},1]$. Moreover, if $N$ is large enough, then $\frac{1}{N+1} \rightarrow 0$. Under this circumstance, the state $\r_N(x)$ will always be free entangled. Next, we give the condition for the biseparability of the D{\"u}r's multipartite state $\r_N(x)$.

\begin{theorem}
	\label{thm:dursep}
	The state $\r_N(x)$ is biseparable if $x \in [0,\frac{1}{2}]$.
\end{theorem}
\begin{proof}
	Let the state 
	\begin{eqnarray}
		\label{eq:spectrald}
		\r_{A_1\cdots A_N}(x):=\r_N(x)=x\ketbra{\Psi_G}{\Psi_G}+\frac{1-x}{2N}\sum_{k=1}^N(P_k+\overline{P}_k),
	\end{eqnarray}
	and  $\a_{A_1\cdots A_N}^{(k)}=\frac{x}{N}\ketbra{\Psi_G}{\Psi_G}+\frac{1-x}{2N}(P_k+\overline{P}_k)$, $N \geq 4$. For simplicity, we do not normalize $\a_{A_1\cdots A_N}^{(k)}$. Further, we have $\r_{A_1\cdots A_N}(x)=\sum_{k=1}^N \a_{A_1\cdots A_N}^{(k)}$.
	
	For the state $\a_{A_1\cdots A_N}^{(k)}$, we bond systems $A_1,\cdots,A_{k-1},A_{k+1},\cdots,A_N$ and denote it as $\overline{A}_{k}$. Define a bijection from $\{ \ket{0,\cdots,0}, \ket{1,\cdots,1} \}$ in $(\bbC^2)^{\otimes (N-1)}$ to the basis $ \{\ket{0}, \ket{1}\}$ in $\bbC^2$ as follows: $ \ket{0,\cdots,0} \rightarrow \ket{0}$, $\ket{1,\cdots,1} \rightarrow \ket{1}$. Then we rewrite the state $\a_{A_1\cdots A_N}^{(k)}$ on $\bbC^2 \otimes \bbC^2$ as 
	\begin{eqnarray}
		\a^{(k)}_{\overline{A}_{k}A_{k}}=\frac{x}{2N}(\ket{00}+\ket{11})(\bra{00}+\bra{11})+\frac{1-x}{2N}(\ketbra{01}{01}+\ketbra{10}{10}). \notag
	\end{eqnarray}
	The partial transpose of $\a^{(k)}_{\overline{A}_{k}A_{k}}$ is
	\begin{eqnarray}
		\label{eq:akpartialtranspose}
		(\a^{(k)}_{\overline{A}_{k}A_{k}})^\Gamma=\frac{x}{2N}(\ketbra{00}{00}+\ketbra{11}{11}+\ketbra{01}{10}+\ketbra{10}{01})+\frac{1-x}{2N}(\ketbra{01}{01}+\ketbra{10}{10}).
	\end{eqnarray} 
	It is known that $\a^{(k)}_{\overline{A}_{k}A_{k}}$ is separable if and only if $\a^{(k)}_{\overline{A}_{k}A_{k}}$ is PPT \cite{PhysRevLett.77.1413, HORODECKI20011}. One can verify that the four eigenvalues of $(\a^{(k)}_{\overline{A}_{k}A_{k}})^\Gamma$ in \eqref{eq:akpartialtranspose} are $\frac{1}{2N},\frac{1-2x}{2N}, \frac{x}{2N},\frac{x}{2N}$. So if $\a^{(k)}_{\overline{A}_{k}A_{k}}$ is separable, we have $x \in [0,\frac{1}{2}]$. That is,  $\a_{A_1\cdots A_N}^{(k)}$ is biseparable if $x \in [0,\frac{1}{2}]$.
	Since $\r_N(x)$ is the convex sum of $\a_{A_1\cdots A_N}^{(k)}$, we obtain that $\r_N(x)$ is biseparable if $x \in [0,\frac{1}{2}]$. This completes the proof.
\end{proof}

Thus the D{\"u}r's multipartite state $\r_N(x)$ is not GE if $x \in [0,\frac{1}{2}]$. Further, the value of $x$ is independent from $N$. So the genuine entanglement of $\ket{\Psi_G}$ is not pretty robust against the noise $\frac{1-x}{2N}\sum_{k=1}^N (P_k+\overline{P}_k)$.

\section{quantum information masking under noise}
\label{sec:masking}
For bipartite systems, an operation $\mathcal{S}$ is defined as quantum information masker if it maps states $\{\ket{a_k}_{A_1}\in \mathcal{H}_{A_1}\}$ to states $\{\ket{\psi_k} \in \mathcal{H}_{A_1} \otimes \mathcal{H}_{A_2}\}$ such that all the reductions to one party of $\ket{\psi_k}$ are identical. For multipartite systems, an operation $\mathcal{S}$ is an $m$-uniform quantum information masker if it maps states $\{\ket{a_k}_{A_1}\in \mathcal{H}_{A_1}\}$ to states 
$\{\ket{\psi_k} \in \otimes_{\ell=1}^n\mathcal{H}_{A_\ell}\}$ such that all the reductions to $m$ parties of $\ket{\psi_j}$ are identical. Moreover, if $m=\lfloor\frac{n}{2}\rfloor $, then this $m$-uniform masking is strong quantum information masking \cite{feishimask}. The action of the masker is a physical process. It can be modeled by a unitary operator $U_{\mathcal{S}}$ on the system $A_1$ plus some ancillary systems $\{ A_2,\cdots,A_n\}$, given by $\mathcal{S}:U_{\mathcal{S}}\ket{a_k}_{A_1}\otimes\ket{b}_{A_2\cdots A_n}=\ket{\psi_k}$. As far as we know, quantum information masking under noise is little studied. In this section, we investigate the quantum information masking under the $n$-partite $d$-dimensional Pauli channel $\cU$ in \eqref{eq:channelU} and the locally unitary channel $\Lambda$ in \eqref{eq:channellambda}, respectively. 

\begin{figure}[!h] 
	\center{\includegraphics[width=1\textwidth]  {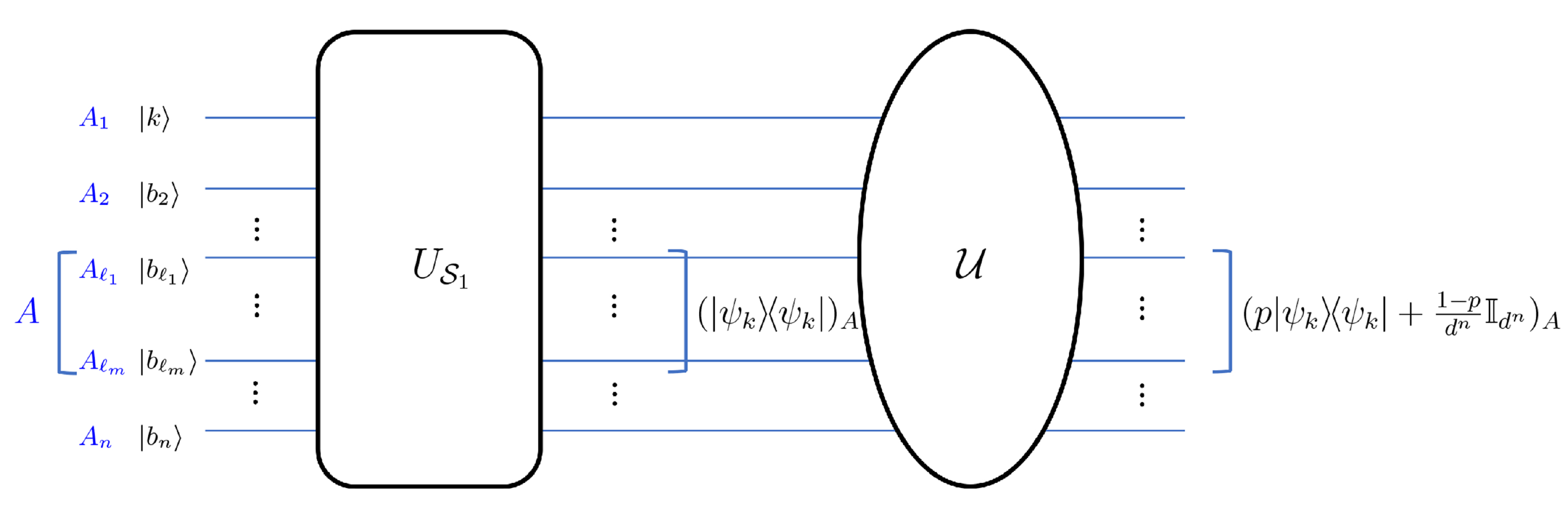}} 
	\caption{\label{fig:channelnoise} Quantum information masking under noise. The state $\ket{k}$ is to be encoded in the $m$-uniform quantum information masking process modelled by $U_{\mathcal{S}_1}$ and the $n$-partite $d$-dimensional Pauli channel $\cU$, where $\ket{b_i}, 2\leq i \leq n$ are the initial states of the ancillary systems. The reduced density matrices $(\ketbra{\psi_k}{\psi_k})_A$ of any $m$ subsystems are the same, where $A=\{A_{\ell_1}, \cdots ,A_{\ell_m}\}$. Then the state $\ketbra{\psi_k}{\psi_k}$ goes through the channel $\cU$, the reduced density matrices $(p\ketbra{\psi_k}{\psi_k}+\frac{1-p}{d^n}\mathbb{I}_{d^n})_A$ of the resulting state $p\ketbra{\psi_k}{\psi_k}+\frac{1-p}{d^n}\mathbb{I}_{d^n}$ are still independent of the encoded state $\ket{k}$. Thus this $m$-uniform masking under noise still works.
	}  
\end{figure}

First, as shown in FIG. \ref{fig:channelnoise}, we consider the quantum information masking and the $n$-partite $d$-dimensional Pauli channel $\cU$. A masker $\mathcal{S}_1: U_{\mathcal{S}_1}\ket{k}_{A_1} \otimes \ket{b}_{A_2 \cdots A_{n}}=\ket{\psi_k}$ masks quantum information contained in $\{ \ket{0},\ket{1},\cdots,\ket{d-1} \} \in \bbC^d$ into $\{ \ket{\psi_0}, \ket{\psi_1}, \cdots, \ket{\psi_{d-1}}  \} \in (\bbC^d)^{\otimes n}$, where $\ket{\psi_k}=\frac{1}{\sqrt{d}}\sum_{j=0}^{d-1}e^{\frac{2\p i}{d}jk}\ket{\underbrace{j,\cdots,j}_n}$, $k=0,1,\cdots,d-1$, and the state $\ket{b}_{A_2 \cdots A_{n}}=\ket{b_2} \otimes \cdots \otimes \ket{b_n}$. One can verify that for all reduced density matrices $(\ketbra{\psi_j}{\psi_j})_{A_{\ell_1}\cdots A_{\ell_m}}$ of $(\ketbra{\psi_j}{\psi_j})_{A_1\cdots A_n}$ for any given $m \in \{1,\cdots,n-1\}$ are $\frac{1}{d}\sum_{j=0}^{d-1}\ketbra{\underbrace{j,\cdots,j}_{m}}{\underbrace{j,\cdots,j}_{m}}$, where $\ell_1,\cdots,\ell_m \in \{1,\cdots,n\}$. They are the same. Thus the messages of $\ket{0}, \cdots, \ket{d-1}$ are strongly masked. Moreover, the masker $\mathcal{S}_1$ is $m$-uniform. Next, the state $\ket{\psi_j}$ goes through the $n$-partite $d$-dimensional Pauli channel $\cU$, we have 
\begin{eqnarray}
	\label{eq:channelpsi}
	&&
	\cU(\ketbra{\psi_k}{\psi_k}) \notag \\
	=&&pK_0\ketbra{\psi_k}{\psi_k}K_0^\dagger+\frac{1-p}{d^n}\sum_{r=1}^{d^n}K_r\ketbra{\psi_k}{\psi_k}K_r^\dagger, \notag \\
	=&&p\ketbra{\psi_k}{\psi_k}+\frac{1-p}{d^n}\mathbb{I}_{d^{n}}.
\end{eqnarray}
The $m$-partite reduced density matrices $(p\ketbra{\psi_k}{\psi_k}+\frac{1-p}{d^n}\mathbb{I}_{d^n})_{A_{\ell_1}\cdots A_{\ell_m}}$ of the state in \eqref{eq:channelpsi} for any given $m \in \{1,\cdots,n-1\}$ are $
\frac{1}{d}\sum_{j=0}^{d-1}\ketbra{\underbrace{j,\cdots,j}_{m}}	{\underbrace{j,\cdots,j}_{m}}+\mathbb{I}_{d^{m}}$. Thus the resulting state is still strongly quantum information masked. Hence, the state in $\{ \ket{0},\ket{1},\cdots,\ket{d-1} \} \in \bbC^d$ can be strongly masked by the masker $\mathcal{S}$ and the $n$-partite $d$-dimensional Pauli channel $\cU$.

Next, for D{\"u}r's multipartite state $\r_N(x)$, we study the quantum information masking and the locally unitary channel $\Lambda$ in \eqref{eq:channellambda}. A masker $\mathcal{S}_2: U_{\mathcal{S}_2}\ket{i}_{A_1} \otimes \ket{b}_{A_2,\cdots,A_{N}}=\ket{\Psi_{G_i}}$ masks quantum information contained in  $\{ \ket{0},\ket{1}\} \in \bbC^2$ into $\{ \ket{\Psi_{G_0}}, \ket{\Psi_{G_1}}\} \in (\bbC^2)^{\otimes N}$, where $\ket{\Psi_{G_i}}=\frac{1}{\sqrt{2}}\sum_{j=0}^{1}(-1)^{ij}\ket{j,\cdots,j}$, $i=0,1$. The reduced density matrices $(\ketbra{\Psi_{G_i}}{\Psi_{G_i}})_{A_{\ell_1}\cdots A_{\ell_m}}$ of $(\ketbra{\Psi_{G_i}}{\Psi_{G_i}})_{A_1\cdots A_n}$ for any given $m \in \{1,\cdots,n-1\}$ are identical, that is, $\frac{1}{2}\sum_{j=0}^{1}\ketbra{\underbrace{j,\cdots,j}_{m}}{\underbrace{j,\cdots,j}_{m}}$, where $\ell_1,\cdots,\ell_m \in \{1,\cdots,n\}$. Hence we have no information about the value of $i$ and it is a strong quantum information masking. Let $\ket{\Psi_{G_i}}$ go through the locally unitary channel $\Lambda$ in \eqref{eq:channellambda}. The state $\ket{\Psi_{G_i}}$ will become 
\begin{eqnarray}
	\label{eq:channelPsi}
	&&
	\Lambda (\ketbra{\Psi_{G_i}}{\Psi_{G_i}}) \notag \\
	=&&x\ketbra{\Psi_{G_i}}{\Psi_{G_i}}+\frac{1-x}{2N}\sum_{k=1}^N(P_k+\overline{P}_k).
\end{eqnarray}
The $(N-1)$-partite reduced density matrices of the resulting state in \eqref{eq:channelPsi} are
\begin{eqnarray}
	&&
	\frac{(N-1)x+1}{2N}\sum_{j=0}^1\ketbra{\underbrace{j,\cdots,j}_{N-1}}{\underbrace{j,\cdots,j}_{N-1}}+\frac{1-x}{2N}\sum_{k=1}^{N-1}(P^\star_k+\overline{P^\star}_k), \notag
\end{eqnarray}
where $P^\star_k$ is the projector onto the product state $\ket{\phi_k}=\ket{0}_{A_1}\ket{0}_{A_2}\cdots\ket{1}_{A_k}\cdots \ket{0}_{A_{N-1}}$, and $\overline{P^\star}_k$ is the projector onto the product state $\ket{\varphi_k}=\ket{1}_{A_1}\ket{1}_{A_2}\cdots\ket{0}_{A_k}\cdots \ket{1}_{A_{N-1}}$. So the resulting state is $1$-uniformly masked.   

The above two examples show that quantum information masking under noise (the channels we constructed) still works. The traditional definition of quantum information masking asks for a map from a pure state to another pure state. Our channels indirectly maps a pure state to a mixed state and the masking effect still works. It is known that noise is not avoidable in nature. Thus quantum information masking under noise plays a key role in quantum secret sharing in experiment. 

\section{DSD of tripartite GE States}
\label{sec:DSD333}

A bipartite system under dephasing will evolve into the dynamics that DSD must precede ESD \cite{DSD_MA}. In this section, we consider the three-qutrit GE state distillable across every bipartition \cite{HalderGES}. We investigate the evolution of it under global (i.e., local and collective) dephasing in \eqref{eq:operator}. It turns out that the state will undergo DSD and become PPT. We study the existence of PPT entanglement using realignment criterion and indecomposable positive map.

Assume that $\ket{\eta_i}=\ket{0}+(-1)^i\ket{1}$, $\ket{\xi_j}=\ket{1}+(-1)^j\ket{2}$, where $i,j=0,1$. Let
\begin{equation*}
\begin{aligned}
\cB=\{&\ket{\psi(i,j)}_1=\ket{0}\ket{\eta_i}\ket{\xi_j}, \; \ket{\psi(i,j)}_2=\ket{\eta_i}\ket{2}\ket{\xi_j}, \; \ket{\psi(i,j)}_3=\ket{2}\ket{\xi_j}\ket{\eta_i},\\
&\ket{\psi(i,j)}_4=\ket{\eta_i}\ket{\xi_j}\ket{0}, \; \ket{\psi(i,j)}_5=\ket{\xi_j}\ket{0}\ket{\eta_i}, \; \ket{\psi(i,j)}_6=\ket{\xi_j}\ket{\eta_i}\ket{2}, \\
&\ket{\psi(0,0)}_1-\ket{\psi(0,0)}_2, \; \ket{\psi(0,0)}_3-\ket{\psi(0,0)}_4, \; 
\ket{\psi(0,0)}_5-\ket{\psi(0,0)}_6,\\
&(\ket{0}+\ket{1}+\ket{2})(\ket{0}+\ket{1}+\ket{2})(\ket{0}+\ket{1}+\ket{2}), \;
i,j=0,1
\}.
\end{aligned}
\end{equation*}
We review the unextendible biseparable base \cite{DSD_MA},  $\cB_1=\cB\setminus \{\bigcup_{\ell=1}^6\ket{\psi(0,0)}_\ell\}$ and 
\begin{eqnarray}
\label{eq:r3}
\rho(0):=\frac{1}{5}\left(\mathbb{I}_3\otimes \mathbb{I}_3\otimes \mathbb{I}_3-\sum_{\ket{\psi}\in{\cB_1}}\ketbra{\widetilde{\psi}}{\widetilde{\psi}}\right),
\end{eqnarray}
where $\ket{\widetilde{\psi}}$ is a normalized state of $\ket{\psi}$. Then $\rho(0)$ is a three-qutrit rank-five GE state, and it is distillable across every bipartite cut \cite{DSD_MA}. One can prove that the three bipartite reduced density matrices $\r_\b(0):=\text{Tr}_\a[\r(0)]$ with $\b \in \{BC,CA,AB\}$ and $\a \in \{A,B,C\}$, respectively, are identical. Hence it suffices to check whether the state $\r(0)$ will undergo ESD and DSD across the cut $\r_{A|BC}(0)$. 

\begin{figure}[!h] 
	\center{\includegraphics[width=10cm]  {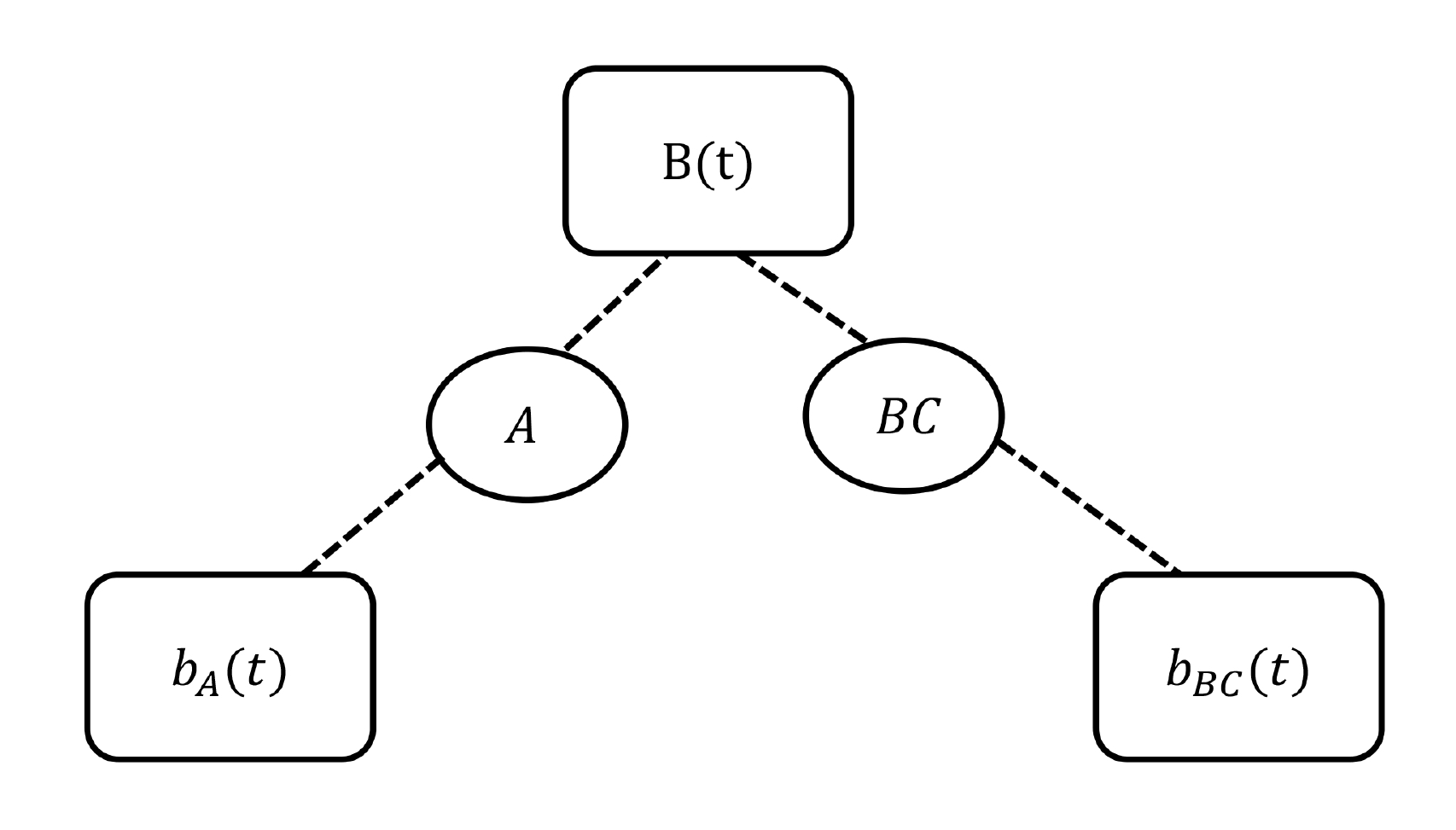}} 
	\caption{\label{fig:dephase} Two systems $A$ and $BC$ are collectively interacting with the stochastic magnetic field $B(t)$ and separately interacting with the stochastic magnetic field $b_A(t)$ and $b_{BC}(t)$.
	}  
\end{figure}

As shown in FIG. \ref{fig:dephase}, the initial state $\r(0)$ interacting with global dephasing only correlates with the three stochastic magnetic fields $b_A(t), b_{BC}(t)$ and $B(t)$ \cite{PhysRevA.76.044101, DSD_SW}. The time-dependent density matrix for the system is $\r(t)=\varepsilon(\r(0))=\sum_\mu K_\mu^\dagger(t)\r(0)K_\mu(t)$, where the Kraus operators $K_\mu(t)$ of the channel $\varepsilon$ representing the influence of statistical noise are completely positive and trace preserving, that is, $\sum_\mu K_\mu^\dagger K_\mu=\mathbb{I}$. By bonding systems $B$ and $C$, We suppose that systems $A$, $BC$ only correlate with the local magnetic fields in FIG. \ref{fig:dephase}. That is, $\r(0)$ interacts with the  local and collective dephasing noise. It can be expressed as 
\begin{eqnarray}
	\label{eq:rl(t)_3}
	&&
	\r_g(t)=\sum_{i=1}^{3}\sum_{j=1}^9 (E_i^{A})^\dagger (F_j^{B})^\dagger \r(0) E_i^AF_j^B,
\end{eqnarray}
where $E_i^A$ and  $F_j^B$ describe the interaction with the local magnetic fields $b_A(t)$ and $b_{BC}(t)$ respectively. The operators $E_i^A$ and  $F_j^B$ are
\begin{eqnarray}
	\label{eq:operator}
	&&
	E_1(t)=\diag(1,\g_A(t),\g_A(t))\otimes \mathbb{I}_9, \notag \\
	&&
	E_2(t)=\diag(0,\omega_A(t),0)\otimes \mathbb{I}_9, \notag \\
	&&
	E_3(t)=\diag(0,0,\omega_A(t))\otimes \mathbb{I}_9, \notag \\
	&&
	F_1(t)=\mathbb{I}_3\otimes\diag(1,\g_B(t),\g_B(t),\dots,\g_B(t)),\notag \\
	&&
	F_2(t)=\mathbb{I}_3\otimes\diag(0,\omega_B(t),0,\dots,0),\notag \\
	&&
	\dots \notag \\
	&&
	F_9(t)=\mathbb{I}_3 \otimes \diag(0,0,\dots,0,\omega_B(t)).
\end{eqnarray}
The time-dependent parameters are $\g_A(t)=\g_B(t)=e^{-\varGamma_{1}t/2}$, $\omega_A(t)=\sqrt{1-\g_A^2(t)}$, $\omega_B(t)=\sqrt{1-\g_B^2(t)}$, and $\varGamma_{1}$ is the dephasing rate.

\begin{figure}[!h] 
	\center{\includegraphics[width=1\textwidth]  {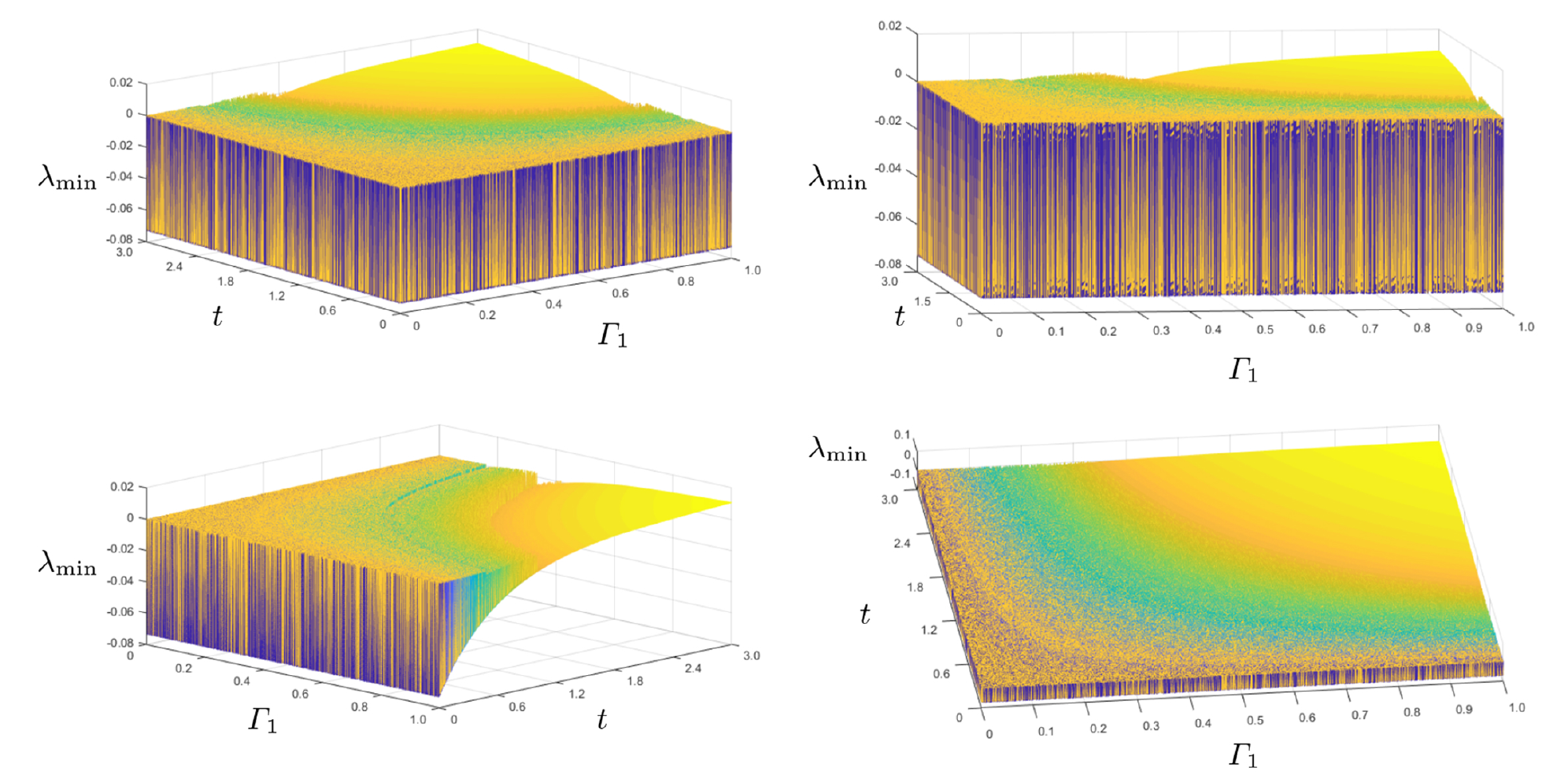}} 
	\caption{\label{fig:eigenvt-4} The three-dimensional plots of the minimum eigenvalue $\lambda_{\min}$ of $\r^\Gamma_g(t)$ in different viewing angles with the local asymptotic dephasing rate $\varGamma_1 \in [0,1]$ and the time $t \in [0,3]$. One can see that $\r^\Gamma_g(t)$ is an NPT state for any $t \in [0,3]$ and $ \varGamma_1 \in [0,0.459)$, or any $t \in [0, 1.377)$ and $ \varGamma_1 \in [0,1]$.
	}  
\end{figure}

\begin{figure}[!h] 
	\center{\includegraphics[width=1\textwidth]  {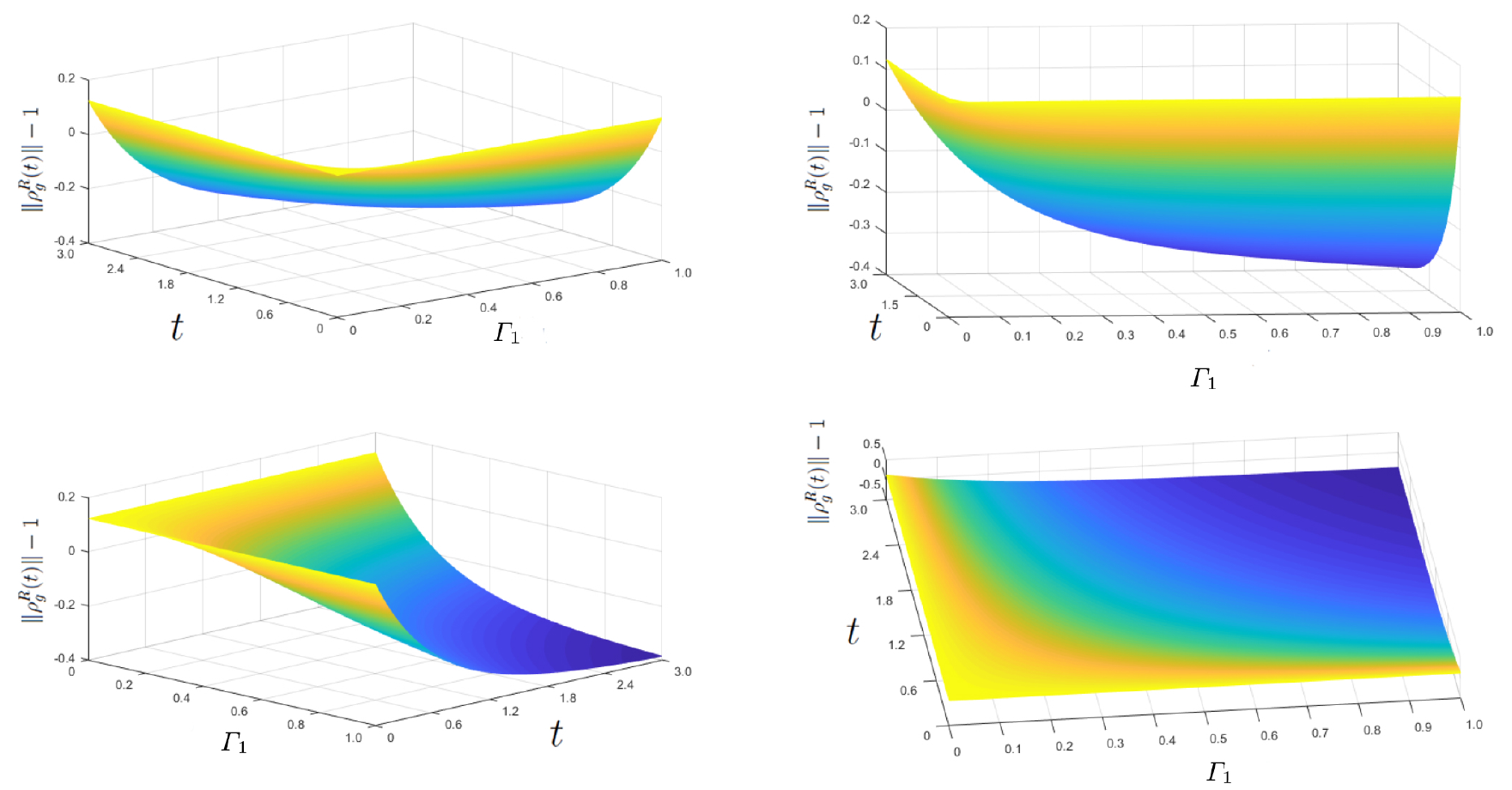}} 
	\caption{\label{fig:realigvt-4} The three-dimensional plots of $\|\r^R_g(t)\|-1$ versus $t \in [0,3]$ and $\varGamma_1 \in [0,1]$ in different viewing angles. We can obtain that $\|\r^R_g(t)\|-1$ is  positive for any $t \in [0,0.186]$ and $\varGamma_1 \in [0,1]$, or any $t \in [0,3]$ and $\varGamma_1 \in [0,0.062]$.
	}  
\end{figure}

Let $\r^\Gamma_g(t)$ be the partial transpose of $\r_g(t)$ in Eq. \eqref{eq:rl(t)_3}. Consider the local asymptotic dephasing rate $\varGamma_1 \in [0,1]$ and the time $t \in [0,3]$. We compute the minimum eigenvalue $\lambda_{\min}(\r^\Gamma_g(t))$ of $\r^\Gamma_g(t)$. Then we find that $\r_g(t)$ is NPT for any $t \in [0,3]$ and $ \varGamma_1 \in [0,0.459)$, or any $t \in [0, 1.377)$ and $ \varGamma_1 \in [0,1]$. So under these circumstances, the state $\r_g(t)$ is entangled. For any $t \in (1.377,3]$ and some $ \varGamma_1 \in (0.459,1]$, the state $\r_g(t)$ is PPT. Thus the state $\r(0)$ will undergo multipartite DSD. In FIG. \ref{fig:eigenvt-4}, we draw the three-dimensional plots of the minimum eigenvalue $\lambda_{\min}(\r^\Gamma_g(t))$ in different viewing angles. Next we use the realignment criterion in Lemma \ref{le:realignment} to compute $\|\r^R_g(t)\|-1$ with $\varGamma_1 \in [0,1]$ and the time $t \in [0,3]$. Similarly, we draw the three-dimensional plots of $\|\r^R_g(t)\|-1$ in different viewing angles in FIG. \ref{fig:realigvt-4}. For any $t \in [0,0.186]$ and $\varGamma_1 \in [0,1]$, or any $t \in [0,3]$ and $\varGamma_1 \in [0,0.062]$, the value of $\|\r^R_g(t)\|-1$ is positive. If the minimum eigenvalue of $\r^\Gamma_g(t)$ and the value of $\|\r^R_g(t)\|-1$ are positive at the same time, then Lemma \ref{le:realignment} implies that $\r(0)$ is PPT entangled. 

In the rest of this section, we choose the local asymptotic dephasing rate $\varGamma_1=1$. Then we can obtain that $\r_g(t)$ will become a PPT state after $t\approx 1.38$. By computing $\|\r^R_g(t)\|-1$, we obtain that if  $t>0.18$, then the value of $\|\r^R_g(t)\|-1$ is positive. So $\r_g(t)$ with $\varGamma_1=1$ can not be distillable under LOCC when $t>1.38$. Next, we use the indecomposable positive map proposed in \cite{Bhattacharya2020GeneratingAD} to detect whether $\r_g(t)$ is PPT entangled. It can detect entanglement in a certain class of two-qutrit PPT entangled states. The one parameter class of linear trace preserving maps $\Lambda_{\a}$ is 
\begin{eqnarray}
\label{eq:lambda}
\Lambda_{\a}=1/(\a+\frac{1}{\a})
\bma
\a(x_{11}+x_{22}) & -x_{12} & -\a x_{13} \\
-x_{21} & \frac{1}{\a}(x_{22}+x_{33}) & -x_{32} \\
-\a x_{31} & -x_{23} & \a x_{33}+\frac{1}{\a}x_{11}
\ema,
\end{eqnarray}
where $X=\bma
x_{11} & x_{12} & x_{13} \\
x_{21} & x_{22} & x_{23} \\
x_{31} & x_{32} & x_{33}
\ema$ is any $3 \times 3$ matrix and $\a \in (0,1]$. Thus, if $\r_g(t)$ in \eqref{eq:rl(t)_3} is PPT and the minimum eigenvalues of $(\mathbb{I}_9\otimes \Lambda_{\a})\r_g(t)$ is smaller than zero, then we can say that $\r_g(t)$ is PPT entangled for some $t$ and $\varGamma_1$. By choosing the local asymptotic dephasing rate $\varGamma_1=1$, we compute the minimum eigenvalues $\lambda_{\min}$ of $(\mathbb{I}_9\otimes \Lambda_{\a})\r_g(t)$ versus time $t \in [0,3]$ and $\a \in (0,1]$.  In FIG. \ref{fig:mapt3}, we plot the three-dimensional plots of minimum eigenvalues $\lambda_{\min}$ of $(\mathbb{I}_9\otimes \Lambda_{\a})\r_g(t)$ in different viewing angles versus time  $t \in [0,3]$ and $\a \in (0,1]$. We can see that for $\a \in (0,1]$, only when $t \in [0,0.666)$, the minimum eigenvalue $\lambda_{\min}$ of $(\mathbb{I}_9\otimes \Lambda_{\a})\r_g(t)$ is negative. However, $t=0.666<1.38$. The numerical results do not show the PPT entanglement of $\r_g(t)$ when $\varGamma_1=1$. Hence, $\r_g(t)$ with the dephasing rate $\varGamma_{1}=1$ will become PPT, thus is not distillable under LOCC. That is, the state $\r(0)$ with local and collective noise in \eqref{eq:operator} when $\varGamma_1=1$ will undergo multipartite DSD. 

\begin{figure}[!h] 
	\center{\includegraphics[width=1\textwidth]  {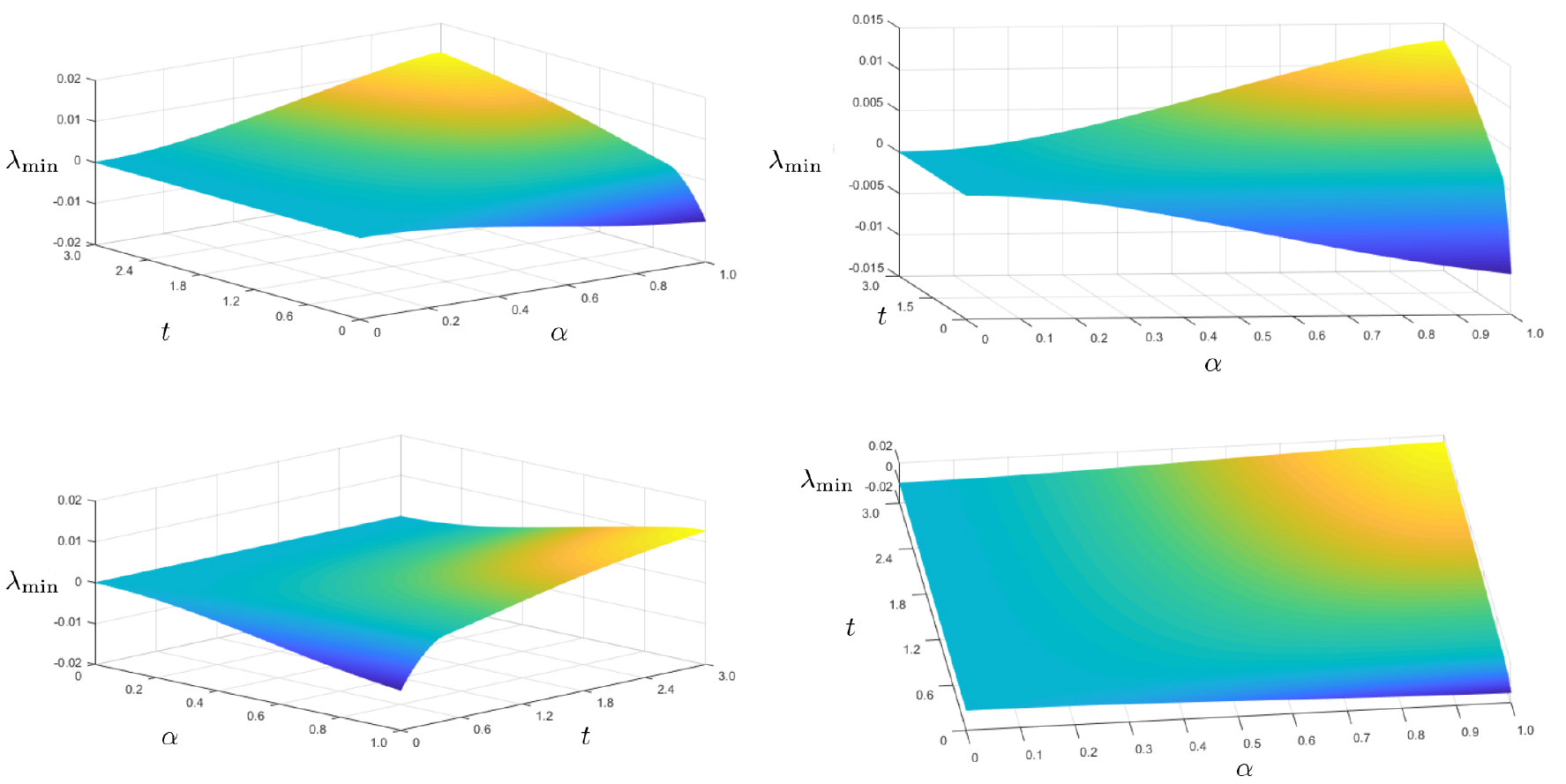}} 
	\caption{\label{fig:mapt3}	The three-dimensional plots of minimum eigenvalues $\lambda_{\min}$ of $(\mathbb{I}_9\otimes \Lambda_{\a})\r_g(t)$ in different viewing angles versus time $t$ and $\a$ for the local asymptotic dephasing rate $\varGamma_1=1$. For $\a \in (0,1]$, the minimum eigenvalue $\lambda_{\min}$ of $(\mathbb{I}_9\otimes \Lambda_{\a})\r_g(t)$ is negative when $t \in (0,0.666)$.		
	}  
\end{figure}

\section{Conclusions}
\label{sec:con}
We constructed the GESD-free states and multipartite DSD-free states by establishing the multipartite high dimensional Pauli channel. We presented the locally unitary channel such that the $N$-partite GHZ state becomes the D{\"u}r's multipartite state. We also studied the quantum information masking under the channel we constructed. Further, we investigated the three-qutrit GE state distillable across every bipartition under global noise. The numerical results showed that it may undergo DSD and become PPT. Our future work is to give more constructions of GESD-free and multipartite DSD-free states from other states such as Werner states as well as their relation to quantum information masking under more general noise.

\section*{Acknowledgments}
\label{sec:ack}	
Authors were supported by the  NNSF of China (Grant No. 11871089), and the Fundamental Research Funds for the Central Universities (Grant Nos. KG12080401 and ZG216S1902).

\bibliographystyle{IEEEtran}

\bibliography{mengyao=sddistillability}

\end{document}